\documentclass[runningheads]{llncs}
\usepackage[T1]{fontenc}
\usepackage{graphicx}
\newcommand{\keysum}{\textsf{keySum}}
\newcommand{\valuesum}{\textsf{valueSum}}
\newcommand{\tagsum}{\textsf{tagSum}}

\usepackage{subcaption}
\usepackage{algorithm}
\usepackage[noend]{algorithmic}
\renewcommand{\emph}[1]{\textbf{\textit{#1}}}

\usepackage{tikz}
\usepackage{xcolor}
\usepackage{amsmath}
\usepackage{amsfonts}
\usepackage{tikz-qtree}
\usepackage[nocompress]{cite}

\def\showryuto{1}

\ifnum\showryuto=1
\newcommand{\ryuto}[1]{{\leavevmode\color{red}Ryuto: #1}}
\else
\newcommand{\ryuto}[1]{}
\fi

\begin{document}
\title{Dynamic Accountable Storage: An Efficient Protocol for Real-time Cloud Storage Auditing}
\titlerunning{Dynamic Accountable Storage}
%
\author{Michael T.~Goodrich\orcidID{0000-0002-8943-191X} \and
Ryuto Kitagawa\orcidID{0009-0000-7329-9590} \and
Vinesh Sridhar\orcidID{0009-0009-3549-9589}}
\authorrunning{M.T. Goodrich {\it et al.}}
%
\institute{University of California, Irvine CA 92697, USA \\
\email{\{goodrich, ryutok, vineshs1\}@uci.edu}}
\maketitle              
\begin{abstract}
Ateniese, Goodrich, Lekakis, Papamanthou, Paraskevas, and Tamassia introduced
the \emph{Accountable Storage} protocol, which 
is a way for a client to outsource their data to a 
cloud storage provider while allowing
the client to periodically perform accountability challenges.
An accountability challenge
efficiently recovers any pieces of data the server has
lost or corrupted, allowing the client to extract the original copies
of the damaged or lost data objects.
A severe limitation of the prior accountable storage scheme of
Ateniese {\it et al.}, however,
is that it is not fully dynamic.
That is, it does not allow a client to freely insert and delete data from the outsourced data set after initializing the
protocol, giving the protocol limited practical use in the real world. 
In this paper, we present \emph{Dynamic Accountable Storage},
which is an efficient way for a client to periodically audit 
their cloud storage while also supporting insert 
and delete operations on the data set. 
To do so, we introduce a data structure,
the \emph{IBLT tree}, 
which allows either the server or 
the client to reconstruct data the server has lost or corrupted 
in a space-efficient way. 


\end{abstract}

\section{Introduction}

Cloud storage providers often advertise the number of ``nines'' of
durability they achieve, such as the ``11 nines'' in
the durability probability of 99.999999999\% advertised by
Amazon S3 for not losing a given
data object in a given year,\footnote{See, e.g., \url{https://aws.amazon.com/s3/storage-classes/.}} 
which implies an expected loss of at most one
object out of 100 billion per year.
Such statements might seem at first to imply that cloud storage data loss
is impossible, until one considers that there are hundreds of trillions
of objects currently being stored in Amazon S3.\footnote{E.g., see
\url{https://aws.amazon.com/blogs/aws/welcome-to-aws-pi-day-2022/}.}
Moreover, such durability statements do not
address data corruption or data loss caused by misuse 
or misconfiguration, e.g., see~\cite{klotz,mellor,pallardy}.
Such durability statements also beg the question of how to
determine whether
one or more of a client's data objects has been lost or corrupted.
For example, 
when data is lost, a client or server may not even realize it and may
also have no ability to recover from the damage. 

We propose an efficient way for a client and/or server 
to audit and recover client data.
Our work is an extension of the \emph{Accountable Storage} 
protocol first described by Ateniese {\it et al.}~\cite{ateniese2017accountable}. 
In this scheme, a client, Alice, outsources her data to
a cloud storage provider, Bob, and she then stores a
small sketch representation of her data set along with some metadata 
in a data structure called an Invertible
Bloom Lookup Table (IBLT)~\cite{eppstein2010straggler,goodrich_invertible_2011}.
At any time, the client can issue an
\emph{accountability challenge}, requiring that the storage provider, Bob, 
send the client a similar encoding of the data set representing everything that
has not been lost. The client can then compare the two IBLTs to
peel out any data blocks that were lost, assess how much the blocks deviate
from their original versions, and demand compensation accordingly. 
Furthermore, Ateniese {\it et al.}~claim that in their scheme 
the server is forced to acknowledge and 
pay for lost data. However, we show that the server can easily thwart their
scheme.

Another caveat in the protocol by Ateniese {\it et al.}
is that their scheme is static, i.e., 
it does not support insert or delete operations.
This severely limits the usefulness of the scheme in
the real world. Our work fixes this limitation; hence, we
call our scheme ``\emph{Dynamic Accountable Storage}.''
Furthermore, our scheme does not rely on forcing the server to pay for lost
data or store large amounts of metadata, as in the original (static)
Accountable Storage scheme~\cite{ateniese2017accountable}.
Instead, our protocol
provides the server with an efficient low-overhead way 
to recover lost data for the client.

\subsection{Other Related Work}
There are many existing schemes that allow a client
to verify that a cloud storage provider is keeping their data intact, but
not with the same levels of efficiency or dynamism as our scheme.
For example,
Provable Data Possession (PDP) schemes
\cite{ateniese2007provable,ateniese2008scalable,curtmola2008mr,erwayDynamicProvableData2008,yu2016identity,jin2016dynamic}
and Proofs of Retrievability (POR) schemes
\cite{cash_dynamic_2017,juels2007pors,shacham2013compact,shi2013practical,armknecht2014outsourced}
use versions of cryptographic tags to verify that data is maintained
correctly. A common technique used in these schemes is \emph{homomorphic tags}
\cite{ateniese2007provable,bellare1998fast}, which enable batch
verification and reduces communication complexity from linear to
constant in the size of the client's data.  
Verifiable Database schemes
\cite{benabbas2011verifiable,catalano2013vector,chen2014new,chen2015verifiable,miao2017efficient,chen2020publicly}
do something similar, except use cryptographic commitment primitives
rather than tags.
Unfortunately, 
dynamic PDP or POR schemes, such as the one by 
Erway, Kupcu, Papamanthou, and Tamassia~\cite{erwayDynamicProvableData2008},
stop at answering whether the
client's data has been corrupted by a cloud service provider. 
In contrast, our
protocol goes further, by recovering the lost data.

Other work, such as that
by Ateniese, Di~Pietro, Mancini, and Tsudik~\cite{ateniese2007provable}
or Shacham and Waters~\cite{shacham2013compact},
implements a type of accountability challenge with dynamic
data, but their schemes require the client to run $O(n)$ expensive
cryptographic operations per update, where $n$ is the size of the
database. Jin {\it et al.}~\cite{jin2016dynamic} solved this recomputation
problem at the expense of having the client store $O(n)$ extra metadata,
somewhat defeating the purpose of the client outsourcing their data.

\subsection{Our Contributions}
In this paper, we describe a scheme for \emph{Dynamic Accountable Storage},
which supports insertions, deletions, and data recovery such
that both the client and server have low overheads in terms of 
additional storage requirements
and running times required to update metadata.
For example, given a parameter, $\delta$, for the number of blocks
our scheme can efficiently tolerate being lost or corrupted, our
scheme requires only $o(n)$ additional space at the server
rather than the $\Theta(n\delta)$ space required by
the original Accountable Storage scheme~\cite{ateniese2017accountable}.

Another difference of our approach from
prior approaches is that our scheme does not 
assume the server is malicious, since
any reputable cloud service company
has an incentive to help the client maintain her data.
As a result,
our work investigates efficient ways for the server to maintain
the client's data with added reliability and to be able to detect
and repair corruptions on the client's behalf, assuming 
that the server is honest-but-curious; hence, we provide a way
for the client, Alice, to encrypt both her keys and their values while still
allowing her to outsource her data with low overheads for both the client
and the server.
We wish to emphasize that an honest-but-curious server does not preclude the need for accountability mechanisms. 
As described above, any cloud storage system of a sufficient size will inevitably produce errors. 
Even if the server is not maliciously corrupting data, identify errant data objects among billions of others is a non-trivial problem.
In addition to cloud storage verification, schemes like ours have applications in version control systems \cite{erwayDynamicProvableData2008}, verifiably-secure logging \cite{crosby2009efficient,tian2017enabling}, and public data auditing \cite{jin2016dynamic,liu2013authorized,ateniese2017accountable}. 



We have two main contributions in this work. 
The first is to show how the
the server in the original Accountable Storage protocol 
of Ateniese {\it et al.}~\cite{ateniese2017accountable} 
can thwart their original scheme so as to recover lost data without paying
the client or revealing that there has been a data loss.
Indeed, we see this as a feature, not a bug, and we build our scheme on
the assumption that the honest-but-curious 
server, Bob, is motivated to recover the client's data
whenever this is possible.

The second contribution of our work is that we
describe an extension to the Accountable Storage protocol 
\cite{ateniese2017accountable} that can support 
insertions and deletions of key-value pairs. 
To allow the server to efficiently detect and repair corruptions 
in the client's data, we introduce a new data structure 
that is maintained at the server, which we call an \emph{IBLT tree}.
This data structure takes $o(n)$ extra space at the server,
compared to the $O(n)$ extra metadata required by the original Accountable Storage protocol.
Although the server must already store $n$ data blocks for the client,
the space savings across a large number of clients are significant.

The remainder of the paper is organized as follows. In
Section~\ref{sec:old}, we briefly review 
the original Accountable Storage protocol 
of Ateniese {\it et al.}~\cite{ateniese2017accountable},
showing how the server can thwart the requirement to pay for lost data.
In Section~\ref{sec:iblt-tree}, we describe
our new data structure, the IBLT tree. Lastly, in Section~\ref{sec:construction},
we provide an overview of our protocol and a formal construction with an efficiency analysis.

\section{Accountable Storage}
\label{sec:old}
In this section, we review prior work on Accountable Storage
by Ateniese {\it et al.}~\cite{ateniese2017accountable}
and describe a flaw that allows a server to avoid acknowledging or paying
for data loss.
Because their scheme was static, however, and ours is dynamic, we
nevertheless describe their scheme in a way that is compatible with
our dynamic approach.

\subsection{Invertible Bloom Lookup Tables} \label{sec:IBLT-def}
We being by reviewing a sketching data structure known as 
the Invertible Bloom Lookup Table (IBLT)~\cite{goodrich_invertible_2011}, 
which is also called the Invertible Bloom Filter~\cite{eppstein2010straggler}.
This data structure is an 
extension of the classic Bloom filter \cite{bloom1970space}
to be a probabilistic key-value store that can be
used in applications such as set difference computation
\cite{eppstein_whats_2011} and straggler identification
\cite{eppstein2010straggler}.
We describe here a slight variation on the IBLT definition, which
is better suited for our protocols.

An IBLT in this context comprises a table, $\mathbf{T}$, 
of $m$ cells, each of which
contains three fields: a \keysum{} field, a \valuesum{} field,
and a \tagsum{} field.
We also have $q$ hash functions, $h_1,h_2,\ldots,h_q$, which respectively map
keys to $q$ distinct cells in the table $\mathbf{T}$.
Initially, all the fields in all the cells of $\mathbf{T}$ are $0$.
We assume that keys and values have fixed bit lengths, where values
can be thought of as data blocks.
For example, in the original accountable storage scheme
of Ateniese {\it et al.}~\cite{ateniese2017accountable},
the keys were the integers, $1,2,\ldots,n$, used to index blocks in
a large file, which were the associated values.
In addition, we have a \emph{tag function}, $\rho$, which is a pseudo-random
hash function that maps key-value pairs to bit strings of sufficient length
that collisions occur with negligible probability.
For example, the accountable storage scheme
of Ateniese {\it et al.}~\cite{ateniese2017accountable}
used the tag function,
\[
\rho(k,v) = \left( h(k) g^v\right)^d \bmod N,
\]
where $h$ is a collision-resistant hash function, $N=pq$ is an RSA
modulus of sufficient bit length, $g$ is a generator of $\mathbf{QR}_N$,
and
$d$ is a secret RSA key (with public-key RSA exponent, $e$) for the client,
Alice, with respect to $N$.
Because the XOR operation, $\oplus$,
is self-cancelling, to insert or delete a key-value pair, $(k,v)$,
with tag, $T_k=\rho(k,v)$,
in an IBLT, $\mathbf{T}$, we use the same algorithm:
\begin{algorithmic}[100]
\STATE ------------------------
\STATE \textsf{Update}$(\mathbf{T},(k,v),T_k)$:
\FOR {$i\leftarrow 1,2,\ldots, q$}
  \STATE $\mathbf{T}[h_i(k)].\keysum{} \leftarrow \mathbf{T}[h_i(k)].\keysum{} \oplus k$
  \STATE $\mathbf{T}[h_i(k)].\valuesum{} \leftarrow \mathbf{T}[h_i(k)].\valuesum{} \oplus v$
  \STATE $\mathbf{T}[h_i(k)].\tagsum{} \leftarrow \mathbf{T}[h_i(k)].\tagsum{} \oplus T_k$
\ENDFOR
\STATE ------------------------
\end{algorithmic}

We can also combine two IBLTs
(which are of the same size and were built using the same hash functions) 
by XORing the respective cells of the two IBLTs,
according to
the following algorithm:
\begin{algorithmic}[100]
\STATE ------------------------
\STATE \textsf{Combine}$(\mathbf{T}_1,\mathbf{T}_2)$:
\FOR {$i\leftarrow 1,2,\ldots, m$}
 \STATE $\mathbf{T}[i].\keysum{} \leftarrow \mathbf{T}_1[i].\keysum{} \oplus \mathbf{T}_2[i].\keysum{}$
 \STATE $\mathbf{T}[i].\valuesum{} \leftarrow \mathbf{T}_1[i].\valuesum{} \oplus \mathbf{T}_2[i].\valuesum{}$
 \STATE $\mathbf{T}[i].\tagsum{} \leftarrow \mathbf{T}_1[i].\tagsum{} \oplus \mathbf{T}_2[i].\tagsum{}$
\ENDFOR
\STATE {\bf return} $\mathbf{T}$
\STATE ------------------------
\end{algorithmic}
This has the effect of
cancelling out the key-value pairs common to the two IBLTs,
while also unioning the key-value pairs that are not in 
the common intersection.
That is, the \textbf{Combine} function computes the unsigned symmetric
difference of the two sets of key-value pairs; hence, if the two sets
of key-value pairs are the same, then the result represents the empty set 
and if the two sets of key-value pairs are disjoint, then the result
is the union of the two sets.




We say that a cell, $i$, in an IBLT, $\mathbf{T}$, is \emph{pure},
if 
\[
\rho(\mathbf{T}[i].\keysum{},\mathbf{T}[i].\valuesum{}) = \mathbf{T}[i].\tagsum{}.
\]
With very high probability, a pure cell should hold just a single 
key-value pair.
Note that by security assumptions, for the particular tag function, $\rho$,
used in Accountable Storage, only the client, Alice, can compute $\rho$
directly.
Nevertheless, note that a server can test whether an IBLT cell, $i$, is pure
by testing whether the following equality holds:
\[
(\mathbf{T}[i].\tagsum)^e \bmod N = \left( h(\mathbf{T}[i].\keysum) g^{\mathbf{T}[i].\valuesum}\right) \bmod N,
\]
where $e$ is the public exponent corresponding to Alice's RSA private key, $d$.

Thus,
if an IBLT, $\mathbf{T}$, is not too full (which we make precise below), 
either the client or the server
can list out the key-value pairs that $\mathbf{T}$
holds by the following \emph{peeling} process:
\begin{algorithmic}[100]
\STATE ------------------------
\STATE \textsf{Peel}$(\mathbf{T})$:
\WHILE {there is a pure cell, $i$, in $\mathbf{T}$}
 \STATE {\bf output} $((\mathbf{T}[i].\keysum,\mathbf{T}[i].\valuesum),
                      \mathbf{T}[i].\tagsum)$
 \STATE \textsf{Update}$(\mathbf{T},(\mathbf{T}[i].\keysum,\mathbf{T}[i].\valuesum),\mathbf{T}[i].\tagsum)$ ~~// delete $(k,v)$
\ENDWHILE
\IF {$\mathbf{T}$ is not empty}
\STATE {\bf output failure}
\ENDIF
\STATE ------------------------
\end{algorithmic}

\begin{lemma}[\hspace*{-4pt}\cite{eppstein_whats_2011,goodrich_invertible_2011}]\label{lem:IBLT-peel}
Let $\mathbf{T}$ be an IBLT constructed using $q$ hash functions that
stores a set of size $\delta$ using at least $(q+1)\delta$ cells and $q$
hash functions, then
\textsf{Peel}($\textbf{T}$) succeeds with probability $1-O(\delta^{-q})$.
\end{lemma}

Thus, the peeling process succeeds with high probability.

\subsection{Accountable Storage}\label{sec:old-review}

In the (static) Accountable Storage scheme~\cite{ateniese2017accountable}, 
the client, Alice, computes and retains an IBLT, $\mathbf{T}_{B}$, 
of size $\Theta(\delta)$,
of her $n$ key-value pairs, where $\delta$ is an upper bound for the expected
number of lost or corrupted blocks that could ever occur at the server, Bob.
In their scheme, Alice numbers her blocks, $1,2,\ldots,n$,
and uses these indices as the keys for the blocks.
In our scheme, however, we support Alice performing ${\sf Put}(k,v,t)$ operations,
which add a new block, $v$, with key, $k$, and tag, $t$,
to Bob's storage, and ${\sf Get}(k)$
operations to retrieve the block and tag associated with the key, $k$.
We also allow Alice to issue a ${\sf Delete}(k,v,t)$ operation, to remove
a block, $v$, with key, $k$, and tag, $t$.
Further, because we view Bob as honest-but-curious,
we allow Alice to encrypt her keys and values with a secret key known only
to her. Bob will therefore not be able to determine the plaintext for
her keys and values without breaking Alice's encryption or using
other means.

Because of its static nature, in the original
Accountable Storage scheme~\cite{ateniese2017accountable}, 
Alice sends all $n$ of her key-value pairs to the server, Bob, 
in a single batch, along with a tag for each such key-value pair.
Bob, in turn, stores the $n$ 
key-value pairs, along with their tags, in his cloud storage and also builds
a complete binary tree, which Ateniese {\it et al.} call a ``segment tree,''
where each internal node stores an IBLT of the key-value pairs stored in
its descendants. 
This segment tree requires $O(\delta n)$ space, which is factor $\delta$
larger than the space needed for Alice's key-value pairs; hence, 
this is a considerable additional burden on Bob, who is likely to
pass this cost on to Alice.
Also, note that the IBLT for the root of this tree 
stores $\mathbf{T}_{\rm all}$, an IBLT containing all blocks.

The client, Alice, can issue 
an accountability challenge~\cite{ateniese2017accountable}
to Bob, which requires Bob to use the segment tree
to construct an IBLT, $\mathbf{T}_{\rm kept}$, for all the key-value
pairs that are not lost or corrupted at the server.
Note that Bob can determine whether any key-value pair is lost if there
no existing value block for a given key (which in original the static
scheme were the integers, $1,2,\ldots,n$). 
Likewise, Bob can determine whether any key-value pair, $(k,v)$, with
tag, $T_k$, is corrupted by testing if
\[
T_k^e \bmod N \not= \left( h(k) g^v\right) \bmod N.
\]
Thus, Bob can determine the set of key-value pairs that were not lost
or corrupted and he can compute $\mathbf{T}_{\rm kept}$ in 
$O(\delta^2\log n)$ time by combining all the IBLTs for nodes in
the segment tree that are immediate children of nodes on paths to 
nodes for lost or corrupted key-value pairs but are themselves not 
on such a path.
Since all the key-value pairs for these IBLTs are disjoint, by construction
of the segment tree, the result of this collection of calls to 
the \textsf{Combine} method results in 
$\mathbf{T}_{\rm kept}$.
In the original Accountable Storage protocol~\cite{ateniese2017accountable},
Bob then sends $\mathbf{T}_{\rm kept}$ to Alice, who computes
\textsf{Peel}(\textsf{Combine}($\mathbf{T}_{B},\mathbf{T}_{\rm kept}$)),
which represents the set of key-value pairs lost or corrupted by the server.
Then Alice
charges Bob for the degree to which this set of now-recovered data 
blocks differs from the corresponding set of key-value pairs in Bob's storage.

\subsection{The Server's Simple Way to Thwart Accountable Storage}
Recall that in an accountability challenge, 
the server, Bob, knows which key-value pairs were lost or corrupted.
But in 
the original Accountable Storage protocol~\cite{ateniese2017accountable},
Bob also knows 
$\mathbf{T}_{\rm all}$, since it is stored in the root of the
segment tree.
Thus, Bob can compute
\textsf{Peel}(\textsf{Combine}($\mathbf{T}_{\rm all},\mathbf{T}_{\rm kept}$)),
and restore all the lost or corrupted key-value pairs without ever paying
Alice.
Moreover, he can do this recovery any time he discovers that a key-value
pair is missing or corrupted. Thus, the server can completely thwart
the 
original Accountable Storage protocol
of Ateniese {\it et al.}~\cite{ateniese2017accountable}.
For this reason, rather than view the server as malicious, which is not
realistic in the real world in any case, we view the server as honest-but-curious
and also interested in recovering a client's data any time it is discovered
to be lost or corrupted.


\newcommand{\rank}{\textsf{rank}_S}
\newcommand{\tree}[1]{\mathcal{T}_{#1}}

\section{IBLT Tree}
\label{sec:iblt-tree}

Note that the server can check if a key-value-tag triple, $(k,v,t)$, 
in its storage is corrupted by checking if 
$
t^e \bmod N = \left( h(k) g^{v}\right) \bmod N
$,
where $e$ is the public exponent corresponding to Alice's RSA private key, $d$.
Thus,
the server can na{\"\i}vely construct an IBLT omitting all corrupted blocks in
$O(n)$ time by iterating through all its $(k,v,t)$ triples and
inserting the non-corrupted triples into the IBLT.
Of course, this na{\"\i}ve
approach is ineffecient in terms of time.
Ateniese {\it et al.}~\cite{ateniese2017accountable} addressed this 
drawback in the static
accountable storage setting by using a data structure they called 
a ``segment tree,''
which increased the space burden on the server to be $O(\delta n)$ while
reducing the time burden for constructing an IBLT of the non-corrupted
blocks to be $O(\delta^2 \log n)$.
We do something similar in the dynamic setting by introducing a 
data structure we call the \emph{IBLT tree},
which provides similar functionality but with much
reduced space.
We begin by introducing a general IBLT tree framework that can be applied to any dynamic tree with bucketed leaf nodes. 
We then apply it to
the protected Cartesian tree
of Bender, Farach-Colton, Goodrich, and Koml{\'o}s~\cite{bender} 
to give concrete time bounds.

\subsection{Overview}
Let $S$ be the set of key-value-tag triples the server stores.
The IBLT tree of the set $S$, $\tree{S},$ is defined as a 
binary search tree (ordered by the keys of the key-value-tag triples) where each
of the leaf nodes contains an IBLT with $\Theta(\beta)$ elements of $S$,
with $\beta\ge 1$ being a tunable parameter.
Each internal node, $v$, in this tree contains the IBLT 
resulting from performing a \textsf{Combine} opeations on the IBLT trees
stored at $v$'s children.

For example, consider an IBLT tree with $\beta = 1$.
This structure would enable the server to efficiently construct an IBLT while
excluding the corrupted blocks, but
this tree would require $O(n \delta)$ space
(recall each node contains an IBLT of size $O(\delta)$),
which requires a factor $O(\delta)$ 
more space to store than the data blocks themselves.
Therefore, we use a larger value of $\beta$, 
as illustrated in Figure~\ref{fig:large-bst}.

\definecolor{babyblueeyes}{rgb}{0.63, 0.79, 0.95}
\definecolor{pastelred}{rgb}{1.0, 0.41, 0.38}

\begin{figure}[hbt]
\vspace*{-12pt}
  \centering
  \tikzset{every tree node/.style={minimum width=2em,draw,circle},
         blank/.style={draw=none},
         edge from parent/.style=
         {draw,edge from parent path={(\tikzparentnode) -- (\tikzchildnode)}},
         level distance=1.5cm}
  \resizebox{0.55 \textwidth}{!}{
  \begin{tikzpicture}
    \Tree
    [.{}
      [.{}
        [.{}
          [.{}
          [.\node[fill=pastelred]{};
            [.\node[fill=babyblueeyes]{}; ]
            [.\node[fill=babyblueeyes]{}; ]
            ]
            [.\node[fill=pastelred]{};
            [.\node[fill=babyblueeyes]{}; ]
            [.\node[fill=babyblueeyes]{}; ]
            ]
          ]
          [.{} 
          [.\node[fill=pastelred]{}; 
          [.\node[fill=babyblueeyes]{}; ]
          [.\node[fill=babyblueeyes]{}; ]
            ]
            [.\node[fill=pastelred]{}; 
            [.\node[fill=babyblueeyes]{}; ]
            [.\node[fill=babyblueeyes]{}; ]
            ]
          ]
          ]
          [.\node[fill=pastelred]{}; 
          [.\node[fill=babyblueeyes]{}; 
          [.\node[fill=babyblueeyes]{}; ]
          [.\node[fill=babyblueeyes]{}; ]
            ]
            [.\node[fill=babyblueeyes]{}; 
            [.\node[fill=babyblueeyes]{};
            ]
            ]
          ]
        ]
        [.{} 
        [.\node[fill=pastelred]{};
          [.\node[fill=babyblueeyes]{}; ]
            [.\node[fill=babyblueeyes]{};
            [.\node[fill=babyblueeyes]{}; ]
            ]
          ]
          [.{} 
          [.{} 
          [.\node[fill=pastelred]{};
          [.\node[fill=babyblueeyes]{}; ]
          [.\node[fill=babyblueeyes]{}; ]
              ]
              [.\node[fill=pastelred]{};
              [.\node[fill=babyblueeyes]{}; ]
              [.\node[fill=babyblueeyes]{}; ]
              ]
            ]
            [.{} 
            [.\node[fill=pastelred]{};
            [.\node[fill=babyblueeyes]{}; ]
            [.\node[fill=babyblueeyes]{}; ]
              ]
              [.\node[fill=pastelred]{};
              [.\node[fill=babyblueeyes]{}; ]
              [.\node[fill=babyblueeyes]{}; ]
              ]
            ]
          ]
        ]
      ]
    \end{tikzpicture}
  }
  \caption{
    An IBLT trees for a set of 43 blocks and $\beta=1$, 
    where each node contains an IBLT of all the blocks in its subtree.
    If we set $\beta = 6$ and require each leaf to store at least
    $\beta / 2$ blocks and at most $\beta$ blocks, then 
    the blue nodes would be removed while the red nodes 
    would become the new leaf nodes.
    If $\beta$ increases further, then 
    the tree becomes more compressed.
  }
  \label{fig:large-bst}
\end{figure}
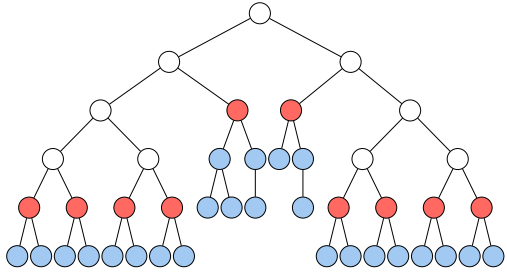

This reduces the space requirement from $O(n \delta)$ to
$O(\delta n / \beta)$,
since the search tree on which the IBLT tree is built
would now contain $O(n / \beta)$ nodes.
Thus, setting $\delta < \beta \leq n$ guarantees sublinear
additional space usage in terms of $n$.

\subsection{IBLT Tree General Framework}
The IBLT tree allows for the following operations:
\textsf{InitIbltTree}$(S),$ \textsf{InsertIbltTree}$(k, v, \tree{S}),$
\textsf{DeleteIbltTree}$(k),$ and
\textsf{ConstructIBLT}$(k^\prime_1, \dots, k^\prime_M, \tree{S}).$
The \textsf{InitIbltTree} operation initializes the IBLT tree with the set $S,$
\textsf{InsertIbltTree} inserts a key-value-tag triple, $(k, v, t)$, 
into the IBLT tree,
\textsf{DeleteIbltTree} deletes the triple with key $k$ from the IBLT tree, and
\textsf{ConstructIBLT} constructs an IBLT of all blocks in the IBLT tree
excluding the blocks with keys in a specified set,
$\{k^\prime_1, \dots, k^\prime_M\}$.
Note that for any dynamic search tree with $\Theta(\beta)$ elements per
leaf node, the IBLT tree functionality can be added with only an $O(\delta)$
factor in space, as previously described, and only a small overhead in time.

\begin{theorem}
  \label{thm:iblt-tree-cost}
  Given a dynamic 
  search tree with $\Theta(\beta)$ elements per leaf node,
  an IBLT tree can be implemented with $O(\delta)$ space overhead
  per node. 
  Given an implementation of the \textsf{Init}, \textsf{Insert}, and
  \textsf{Delete}, operations on the dynamic search tree, we can
  implement an IBLT tree with the following time complexities for each
  operation.
  Let
  $L$ and $I$ be the number of leaf and internal nodes respectively,
  $\Delta_L$  and $\Delta_I$ be the number of leaf and internal nodes
  which are changed by an insertion or deletion operation,
  $H$ be the height of the tree,
  $c$ be the number of children per internal node,
  and
  $M$ be the number of blocks to exclude in \textsf{ConstructIBLT}.
  \begin{enumerate}
    \item The \textsf{InitIbltTree} operation takes 
      $O(T_\textsf{Init} + L\beta + I c \delta)$ time.
    \item The \textsf{InsertIbltTree} and \textsf{DeleteIbltTree} operations 
      take
      $O(T_\textsf{Insert} + \Delta_L \beta + \Delta_I c \delta)$
      and
      $O(T_\textsf{Delete} + \Delta_L \beta + \Delta_I c \delta)$
      time.
    \item The \textsf{ConstructIBLT} operation takes
      $O(H M \delta + \delta \beta)$ time.
  \end{enumerate}
\end{theorem}

The details for the proof of Theorem \ref{thm:iblt-tree-cost} are left in
Appendix \ref{sec:iblt-tree-construction} in the interest of space.

\subsection{Protected Cartesian IBLT Tree}\label{sec:protected-cartesian}
The protected Cartesian tree by Bender, Farach-Colton, Goodrich and
Kolm{\'o}s~\cite{bender} is an example dynamic search tree with $\Theta(\beta)$
elements per leaf node, which we use to illustrate the IBLT tree framework.

We first review the basics of the protected Cartesian tree.
Let $h: k \rightarrow [0, 1]$ be a random hash function that maps
a key $k$ to a real number between $0$ and $1.$
The cartesian tree is a binary search tree
where the nodes adhere to the
additional constraint of being a heap with respect to $h$.
This means with respect to the IBLT framework, $c = 2$.

The protected Cartesian tree introduces the \emph{protected region} of some set
$S,$ defined as
$P(S) =
\{ (k, v) \in S \mid \rank{}(k) \leq \beta / 2 \}
\cup
\{ (k, v) \in S \mid \rank{}(k) > n - \beta / 2 \}
$.
\footnote{
  $\rank{}(k)$ is a function that returns
  $|\{ (k^\prime, v^\prime) \in S \mid k^\prime < k \}|$.
}
Naturally the unprotected region is defined as $U(S) = S \setminus P(S).$

The protected Cartesian tree is constructed by finding the key-value pair
$(k, v)$ within the protected region $P(S)$ with the smallest $h(k)$ value.
This chosen key-value pair becomes the pivot value for the root of the tree.
$S$ is then split into all key-value pairs less than $k,$
$S_1 = \{ (k^\prime, b^\prime) \in S \mid k^\prime < k \}$
and all pairs greater than $k,$
$S_2 = \{ (k^\prime, b^\prime) \in S \mid k^\prime > k \}$.
The left and right subtrees are then recursively constructed on $S_1$ and $S_2$
respectively.
If the protected region is empty, in other words $|S| \leq \beta,$ then the
tree consists of a single leaf node containing $S.$
Since a pivot may only be chosen from $U(S),$ each leaf must contain at
least $\beta / 2$ elements,
and as stated previously leaves exist only if their size is $\leq \beta.$
Thus, the protected regions enable us to percolate $\Theta(\beta)$-sized groups
of blocks in sorted order down to the leaves.

Since each leaf node contains $\Theta(\beta)$ unique elements, then we know
at most there are $O(n / \beta)$ leaf nodes.
And using a well known bound on the number of internal nodes in a binary tree,
we know there must be $O(n / \beta)$ internal nodes as well.

\begin{lemma}(Bender et al.~\cite{bender})
  \label{lem:protected-cartesian-construction}
  The protected Cartesian tree can be constructed on a set $S$ with size $n$
  in $O(n)$ time.
\end{lemma}

From Lemma
\ref{lem:protected-cartesian-construction},
Theorem~\ref{thm:iblt-tree-cost},
and some basic facts about the protected Cartesian tree,
we get the following corollary.

\begin{corollary}\label{cor:InitIbltTree-runtime}
  The \textsf{InitIbltTree} operation on the protected Cartesian tree on set
  $S$ of size $n$ takes $O(n)$ time.
  \footnote{
    Recall that $\beta > \delta,$ meaning the $O(\delta n / \beta)$ term from
    constructing the internal nodes is upper bounded by $O(n).$
  }
\end{corollary}

\begin{lemma}\label{lem:InsertIbltTree-runtime}
  The time to perform the \textsf{InsertIbltTree} or \textsf{DeleteIbltTree}
  operation on the protected Cartesian tree in expectation is
  $O(\delta \log(n / \beta) / \beta)$.
\end{lemma}

And finally, the runtime of \textsf{ConstructIBLT} operation falls readily from
Lemma~\ref{lem:construct-iblt} and Theorem~\ref{thm:iblt-tree-cost} (both proved in Appendix~\ref{sec:iblt-tree-construction}),
to give us Corollary~\ref{cor:construct-iblt-protected-cartesian}.

\begin{corollary}
  \label{cor:construct-iblt-protected-cartesian}
  The time to perform the \textsf{ConstructIBLT} operation on the protected
  Cartesian tree in expectation is
  $O(\delta M \log(n / \beta) + \delta \beta)$.
\end{corollary}

\section{Our Dynamic Accountable Storage Construction}\label{sec:construction}

\subsection{Model}

As described in Section~\ref{sec:old}, the purpose of the original Accountable Storage protocol is to allow the client to determine whether the server has lost or corrupted her data. In addition to allowing the client to add and remove data from her database, we also refined the trust model to incorporate the server as an ally rather than an adversary. Formally, 
we use an \emph{honest-but-curious} model~\cite{honest} for the server---the client trusts the server not to actively tamper with her data, but still encrypts her keys and blocks to keep them secret from the server. 
Thus, rather than have the client incur the communication and time costs of auditing her own data, our protocol allows the server to audit the data on her behalf. We also allow the client to either audit a specific set of values or the entire database if she chooses to do so.  

The result of an audit is either the set of corrupted blocks and their recovered counterparts, or some failure state. As we will see below, a failure state is only likely to occur if the server has corrupted the client's data past a specified threshold. Thus, if audits happen periodically, it is unlikely the client will ever have to issue a complaint to the server for missing data. Below we present the seven functions that constitute Dynamic Accountable Storage.

\begin{definition}[Dynamic $\delta$-AS scheme]\label{def:DAS}
    A dynamic $\delta$-AS scheme has the following polynomial-time functions:
\begin{enumerate}

  \item $(pk, sk, T, state, \mathcal{T}) \leftarrow \mathsf{Setup}(B, K, \delta, 1^\tau, \beta)$. Takes as input the client's set of encrypted blocks and keys $(B, K)$, the parameter $\delta$, the security parameter $\tau$, and the server parameter $\beta$. Returns a public and secret key $(pk, sk)$, a set of tags $T$, the client state, and the server's IBLT tree $\mathcal{T}$. 

\item $(state') \leftarrow \mathsf{Put}(k, b, state)$. Takes the new key-block pair and the client's state. The client inserts the key-block pair into their IBLT and has the server do the same.

\item $\{b, \mathtt{failure}\} \leftarrow \mathsf{Get}(k)$. Takes the key of the block the client requests. The server sends the corresponding block $b$ to the client or they both output $\mathtt{failure}$ if the server cannot retrieve it.

\item $\{state', \mathtt{failure}\} \leftarrow \mathsf{Delete}(k, state)$. Takes the key whose block the client wants to delete and the client's state. The client and server delete this block from their states.

\item $\{B', \mathtt{reject}\} \leftarrow \mathsf{ClientAudit}(K', state)$. Takes a set of keys $K'$ the client wishes to audit and the client's state. Returns the corresponding correct blocks or $\mathtt{reject}$, if the blocks could not be recovered correctly. 

\item $\{B', \mathtt{failure}\} \leftarrow \mathsf{ServerAudit}(K', \mathcal{T})$. Takes a set of keys the server wishes to audit and the server's IBLT tree. Returns the corresponding correct blocks or $\mathtt{failure}$, if the blocks could not be recovered correctly. 

\item $\{\mathcal{L}, \mathtt{reject}\}\leftarrow \mathsf{AccountabilityChallenge}(state)$. Takes the client's state. Performs an audit of the entire data base. The client returns a set $\mathcal{L}$ of blocks that the server has corrupted, or $\mathtt{reject}$ if the client could not recover the corrupted blocks. 

\end{enumerate}
\end{definition}

Since the server is an honest-but-curious party in the protocol, 
and we do not consider the risk of data leakage due to the access pattern
for Alice's interactions with her outsourced data,
we have no need for a notion of security beyond having
Alice encrypt her keys and values using a secret key known only to her. 
Below we define our correctness and efficiency goals.

\begin{definition}[Dynamic $\delta$-AS scheme correctness]\label{def:DAS-correctness}

  Let $\mathcal{P}$ be a dynamic $\delta$-AS scheme as defined in Definition~\ref{def:DAS}. $\mathcal{P}$ is correct if 
for all $\tau \in \mathbb{Z}$,
for all $\{pk, sk\}$, tags $T_k$, and states output by $\mathsf{Setup}$,
and for all changes to $state$ made by $\mathsf{Put}$ and $\mathsf{Delete}$, assuming the server has corrupted no more than $\delta$ blocks, the following statements hold:


\begin{enumerate}
    \item For all sets $K^* \subset K$ such that $|K^*| \leq \delta$ input into $\mathsf{ClientAudit}$ and $\mathsf{ServerAudit}$, the probability that the corresponding blocks $B'$ are output is $1 - 1/\text{poly}(\delta)$. 
    \item For all valid $k \in K$, $\mathsf{Get}$ and $\mathsf{Delete}$ succeed with probability $1 - 1/\text{poly}(\delta)$.
    \item For all sets of corrupted blocks $\mathcal{L} \subset B'$ such that $|\mathcal{L}| \leq \delta$, $\mathsf{AccountabilityChallenge}$ outputs $\mathcal{L}$ with probability $1 - 1/\text{poly}(\delta)$.

\end{enumerate}
\end{definition}

\begin{theorem}[Dynamic $\delta$-Accountable Storage]\label{thm:main}
    Let $n$ be the number of blocks. Let the protected Cartesian Tree by Bender {\it et al.} \cite{bender} be the base of the IBLT tree implementation. For all $\delta, \beta \leq n$, there exists a dynamic $\delta$-AS scheme such that
    (1) It is correct according to Definition~\ref{def:DAS-correctness};
    (2) Space usage is $O(\delta)$ at the client and $O(n + \delta n /\beta)$ at the server; 
    (3) $\mathsf{Setup}$ takes $O(n)$ time at both the client and the server; 
    (4) $\mathsf{Put}$ takes $O(1)$ time at the client and $O(\delta\log(n/\beta) / \beta)$ time at the server;
    (5) $\mathsf{Get}$ and $\mathsf{Delete}$ take $O(1)$ time at the client and $O(\delta\log(n/\beta) + \delta\beta)$ time at the server;
    (6) $\mathsf{ClientAudit}$ takes $O(\delta\log n)$ time at the client, $O(\delta^2\log(n/\beta) + \delta\beta)$ time at the server, and has a proof size of $O(\delta)$; 
    (7) $\mathsf{ServerAudit}$ takes no time at the client and $O(\delta^2\log(n/\beta) + \delta\beta + \delta\log n)$ time at the server;
    (8) $\mathsf{AccountabilityChallenge}$ takes $O(\delta\log n)$ time at the client, $O(\delta^2\log(n/\beta) + \delta\beta)$ time at the server, and has a proof size of $O(\delta)$. 
\end{theorem}

Now we will provide a formal construction and prove Theorem~\ref{thm:main}.

\subsection{Construction}

Let $N = pq$ be an RSA modulus with $2\tau$ bits. Let $e$ be a random prime and $d$ be a number such that $de = 1 \mod \phi(N)$. Let $\mathcal{H}$ be a set of $q$ hash functions that maps $\{0, 1\}^*$ to $t$, the size of the IBLT. Let $g$ be a generator of $\mathbf{QR}_N$. Let $h$ be a collision-resistant hash function that maps $\{0, 1\}^*$ to $\mathbf{QR}_N$, modeled as a random oracle. The client begins with an encrypted set of keys $K$ and blocks $B$ that she encrypts with a method of her choice. 

\begin{enumerate}

  \item $(pk, sk, T, state, \mathcal{T}) \leftarrow \textsf{Setup}(B, K, \delta, 1^\tau, \beta)$.  Set $pk = (N, e, \mathcal{H}, g, h)$. Set $sk = d$. The client generates an IBLT $\mathbf{T}_B$ of size $O(\delta)$. For each key-block pair $(k, b)$, she computes the tag $T_k = (h(k) g^b)^d \mod N$ and then runs $\textsf{Update}(\mathbf{T}_B, (k, b), T_k)$ on all of her blocks. She sets $state = (K, \mathbf{T}_B)$ and sends the server her keys, blocks, and tags $(K, B, T)$. The server runs \textsf{InitIbltTree} on $(K, B, T, \beta)$ to produce its IBLT tree $\mathcal{T}$ using buckets of size $\beta$.

  \item $(state') \leftarrow \textsf{Put}(k, b, state)$. The client generates a new tag $T_k = (h(k) g^b)^d \mod N$ and runs $\textsf{Update}(\mathbf{T}_B, (k, b), T_k)$ to insert the block into her IBLT. The client then sends $(k, b, T_k)$ to the server, which executes \textsf{InsertIbltTree} on $(k, b, T_k)$ to insert the block into $\mathcal{T}$.

  \item $\{b, \texttt{failure}\} \leftarrow \textsf{Get}(k)$. The client requests the block corresponding to $k$ from the server. If the server has $b$, then it sends it to the client. If not, it attempts to run $\textsf{ServerAudit}({k}, \mathcal{T})$. If that returns $b$, it sends that to the client, otherwise it returns $\texttt{failure}$. 

If the client receives $b$, she outputs $b$. Otherwise, she returns \texttt{failure}.

\item $\{state', \texttt{failure}\} \leftarrow \textsf{Delete}(k, state)$. The client executes $\mathsf{Get}(k)$ to retrieve the block $b$.

  If the server finds $b$ during \textsf{Get}, it executes \textsf{DeleteIbltTree} on $(k, b, T_k)$ to remove the block from $\mathcal{T}$. Otherwise it outputs \texttt{failure}.

If the client outputs $b$ from $\mathsf{Get}$, she computes tag $T_k = (h(k) g^b)^d \mod N$ and runs $\textsf{Update}(\mathbf{T}_B, (k, b), T_k)$, which now deletes $b$ from her IBLT. Otherwise, she returns \texttt{failure}.

\item $\{B', \texttt{reject}\} \leftarrow \textsf{ClientAudit}(K', state)$. The client sends $K'$ to the server, and the server executes the \textsf{ConstructIBLT} function on $\mathcal{T}$ to compute an IBLT $\mathbf{T}_K$ that contains all blocks except the ones that correspond to the keys $K'$. The client receives $\mathbf{T}_K$ and performs $\mathsf{Combine}(\mathbf{T}_B, \mathbf{T}_K)$ to produce the IBLT $\mathbf{T}_L$. She then performs \textsf{Peel} on $\mathbf{T}_L$ to extract $B'$. If \textsf{Peel} fails, she outputs \texttt{reject}. Otherwise she outputs the set of blocks $B'$. This set will contain both the correct blocks that correspond to $K'$ as well as any blocks the server had corrupted.

\item $\{B', \texttt{failure}\} \leftarrow \textsf{ServerAudit}(K', \mathcal{T})$. Similar to $\mathsf{ClientAudit}$ except the server itself takes the IBLT difference $\mathbf{T}_L$ by executing $\mathsf{Combine}(\mathbf{T}_\textsf{all}, \mathbf{T}_K)$, where $\mathbf{T}_K$ is constructed as above and $\mathbf{T}_{\textsf{all}}$ is the IBLT stored at the root of the server's IBLT tree which contains all blocks. The server can perform $\textsf{Peel}$ on $\mathbf{T}_L$ to recover $B'$.

\item $\{\mathcal{L}, \texttt{reject}\}\leftarrow \textsf{AccountabilityChallenge}(state)$. The client requests an accountability challenge to the server, and the server compiles all the blocks $B$ that it believes belong to the client. Executing the \textsf{ConstructIBLT} function of their IBLT tree $\mathcal{T}$, the server constructs an IBLT $\mathbf{T}_K$ from them. The client computes $\mathbf{T}_L$ using $\mathsf{Combine}(\mathbf{T}_B, \mathbf{T}_K)$. She then uses \textsf{Peel} to extract at most $\delta$ blocks $\mathcal{L}$ that remain in $\mathbf{T}_L$. These blocks represent the corrupted blocks held by the server as well as the correct blocks that correspond. If \textsf{Peel} fails, she outputs \texttt{reject}. Otherwise she outputs the set $\mathcal{L}$.

\end{enumerate}

\subsection{Analysis}

Proofs of the following lemmas and theorems are deferred to Appendix~\ref{sec:appendix-proofs}. Together they imply that Theorem~\ref{thm:main} is true.

\begin{theorem}\label{thm:correctness}
    The Dynamic Accountable Storage construction given in Section~\ref{sec:construction} is correct according to Definition~\ref{def:DAS-correctness}.
\end{theorem}

\begin{lemma}\label{lem:space}
    Space usage is $O(\delta)$ at the client and $O(n + \delta n /\beta)$ at the server.
\end{lemma}

\begin{lemma}\label{lem:setup}
$\mathsf{Setup}$ takes $O(n)$ time at both the client and server. 
\end{lemma}

\begin{lemma}\label{lem:insert-remove}
    $\mathsf{Put}$, $\mathsf{Get}$, and $\mathsf{Delete}$ take $O(1)$ time at the client. $\mathsf{Put}$ takes $O(\delta\log(n/\beta) / \beta)$ time at the server. $\mathsf{Get}$ and $\mathsf{Delete}$ take $O(\delta\log(n/\beta) + \delta\beta)$ time at the server.
\end{lemma}

\begin{lemma}\label{lem:audit}
    $\mathsf{ClientAudit}$ and $\mathsf{AccountabilityChallenge}$ take $O(\delta \log n)$ time at the client and $O(\delta^2\log(n/\beta) + \delta\beta)$ time at the server with a proof size of $O(\delta)$. $\mathsf{ServerAudit}$ takes no time at the client and $O(\delta^2\log(n/\beta) + \delta\beta + \delta\log n)$ time at the server.
\end{lemma}

Instantiating $\beta = \delta$ gives us a protocol whose space at the server is $O(n)$, as opposed to $O(n + n\delta)$ space at the server in the original Accountable Storage protocol \cite{ateniese2017accountable}. With this instantiation, $\mathsf{Put}$ takes $O(\delta \log(n/\delta))$ time in the IBLT tree versus $O(\delta^2\log n)$ time in the segment tree. In addition, $\mathsf{Get}$ and $\mathsf{Delete}$ take $O(\delta^2\log(n/\delta) + \delta^2)$ time in the IBLT tree versus $O(\delta^2 \log n)$ time in the segment tree. Similar performance improvements happen for the audit functions. Thus, the IBLT performs better in both space and time complexity in this configuration. We can further reduce the size of the IBLT tree by increasing the size of $\beta$. However, this incurs an increase in the runtime for $\mathsf{Get}$, $\mathsf{Delete}$, and the auditing functions. The server owner can choose an appropriate value of $\beta$ in accordance with their specific capabilities.




\section{Discussion}
In this work, we observed that the original Accountable Storage protocol of Ateniese, Goodrich, Lekakis, and Papamanthou \cite{ateniese2017accountable} had limited practical use due to being a static protocol and additionally allowed the server to circumvent punishment from the client by using its metadata to recover its invalid data. With the knowledge that cloud storage providers would not be in business if they maliciously corrupted their clients' data, we resolved this contradiction by modeling the server as the client's ally rather than her adversary. We also introduced the IBLT tree, a more efficient version of the segment tree presented in the original work by Ateniese {\it et al.} \cite{ateniese2017accountable}, decreasing the server's burden in implementing accountability challenges. 

In future work, it would be beneficial to exploring how the Dynamic Accountable
Storage protocol can be optimized for larger groups of clients.
For example, if a group of users stores similar data,
a technique to share portions of an IBLT tree may further improve space
efficiency.
Additionally, we wish to provide empirical comparisons in proof size and
computation time between this protocol,
the original Accountable Storage protocol,
and other protocols that guarantee some notion of accountability
\cite{erwayDynamicProvableData2008,jin2016dynamic,juels2007pors,shi2013practical,armknecht2014outsourced,catalano2013vector,chen2015verifiable,chen2020publicly}.

\begin{credits}
\subsubsection{\ackname} 
This work was supported in part by NSF grant 2212129.

\subsubsection{\discintname}
The authors have no competing interests. 
\end{credits}
%
%

\bibliographystyle{splncs04}
\bibliography{citations}

\appendix
\section{Proofs for Dynamic Accountable Storage Construction}\label{sec:appendix-proofs}

\begin{proof}[of Theorem~\ref{thm:correctness}]

    Notice that the only non-deterministic process that affects correctness is the IBLT \textsf{Peel} algorithm used in $\mathsf{ClientAudit}$, $\mathsf{ServerAudit}$, and \\$\mathsf{AccountabilityChallenge}$. The remaining subprocesses are either deterministic, such as \textsf{Update} and \textsf{Combine}, or have randomness that affects runtime but not correctness, as in the IBLT tree update functions. Thus, the three statements of Definition~\ref{def:DAS-correctness} rely on the fact that the client's state is always correct given the updates she has issued and that the IBLT peeling process succeeds with high probability given that there are no more than $\delta$ items in the IBLT.

    It is trivially true that the client updates her state correctly in $\mathsf{Put}$. For $\mathsf{Delete}$, the client must rely on the server to send her the correct block $b$ and its key $k$ using $\mathsf{Get}$. Because the server is honest, it can use the following equation to determine if it has the correct block:
    \[
    T_k^e \bmod N \not= \left( h(k) g^b\right) \bmod N.
    \]

    If $b$ is corrupted and, as we have assumed, the server has corrupted no more than $\delta$ blocks, then as described in the protocol it can invoke \textsf{ServerAudit}. In this case, as shown in Lemma~\ref{lem:IBLT-peel}, the resulting IBLT produced by the audit can be peeled successfully with probability $1 - O(\delta^{-q})$, where $q$ is the number of hash functions used in the IBLT. Even if the peeling process fails, the server will report that to the client. Thus, we have shown that \textsf{Put}, \textsf{Get}, and \textsf{Delete} run successfully with high probability, and even upon failure never cause the client's state to be updated incorrectly.

    As for showing that the algorithms succeed with high probability, notice that we have already proved this for $\mathsf{ServerAudit}$ by citing Lemma~\ref{lem:IBLT-peel}. The remaining two auditing functions $\mathsf{ClientAudit}$ and $\mathsf{AccountabilityChallenge}$ have a similar structure to $\mathsf{ServerAudit}$ and proving that they succeed under the assumptions of Definition~\ref{def:DAS-correctness} is symmetrical. 
\qed
\end{proof}

For the following proofs, we assume that the base of the IBLT tree is the protected Cartesian tree of Bender {\it et al.} \cite{bender} described in Section~\ref{sec:iblt-tree}.

\begin{proof} [of Lemma~\ref{lem:space}]
    The client stores her IBLT $\mathbf{T}_B$, which is of size $O(\delta)$. We do not count her keys $K$ towards her space usage because in any functional outsourced storage system, she would need some way to access her data.

    Along with the client's $n$ blocks, the server stores an IBLT tree. We observed in Section~\ref{sec:protected-cartesian} that there are $O(n/\beta)$ nodes in the protected Cartesian tree. The server must store an IBLT of size $O(\delta)$ in each node, so space usage of the IBLT tree is $O(\delta n / \beta)$. Thus total space usage is $O(n + \delta n / \beta)$.
\qed
\end{proof}

\begin{proof}[of Lemma~\ref{lem:setup}]
Generating $n$ tags and inserting $n$ items into an IBLT using \textsf{Update} takes $O(n)$ time at the client. 

Upon receiving the client's keys, blocks, and tags, the server runs \textsf{InitIbltTree}, which takes $O(n)$ time according to Corollary~\ref{cor:InitIbltTree-runtime}. 
\qed
\end{proof}

\begin{proof} [of Lemma~\ref{lem:insert-remove}]
    For $\mathsf{Put}$, the client simply has to create a tag and run \textsf{Update} once on their IBLT with their new key and block. This takes constant time. Likewise, in $\mathsf{Get}$, the client simply has to send the request and receive the block, which takes constant time. For $\mathsf{Delete}$, the client waits for the server to return their block. If the server does, then the client recreates the tag to verify and then runs \textsf{Update} to peel it from her IBLT. If the server returns \texttt{failure}, then the client immediately returns. Both these situations take constant time at the client. 

    For the server, $\mathsf{Put}$ has the server invoke \textsf{InsertIbltTree}. According to Lemma~\ref{lem:InsertIbltTree-runtime}, this takes $O(\delta\log(n/\beta)/\beta)$ time.

    In $\mathsf{Get}$, there are two cases for the server. Either the server has the block, in which case they simply verify it and return it to the client in constant time. If the server suspects the block is corrupted, it first audits itself. Because we are auditing a single key, we get better bounds. Running \textsf{ConstructIBLT} with $M=1$ only takes $O(\delta \log(n/\beta) + \delta\beta)$ according to Corollary~\ref{cor:construct-iblt-protected-cartesian}. Similarly, running \textsf{Peel} with a single element in the IBLT just requires the server to scan through the IBLT once, taking $O(\delta)$ time. Thus, the worst case time for the server in \textsf{Get} is $O(\delta \log(n/\beta) + \delta\beta)$.

    In \textsf{Delete}, the server executes their response to \textsf{Get} as a subprotocol followed by invoking \textsf{DeleteIbltTree}. According to Lemma~\ref{lem:InsertIbltTree-runtime}, the deletion takes $O(\delta\log(n/\beta)/\beta)$ time. Thus, total runtime for \textsf{Delete} is $O(\delta \log(n/\beta) + \delta\beta)$.
\qed
\end{proof}

\begin{proof}[of Lemma~\ref{lem:audit}]

In $\mathsf{ClientAudit}$ and $\mathsf{AccountabilityChallenge}$, the server first runs \textsf{ConstructIBLT} to build the proof $\mathbf{T}_K$. According to Corollary~\ref{cor:construct-iblt-protected-cartesian}, this takes $O(\delta^2\log(n/\beta) + \delta\beta)$ time when $M$ is set to $\delta$. The proof is an IBLT of size $O(\delta)$, so the proof size must be $O(\delta)$. 

After receiving $\mathbf{T}_K$, the client performs \textsf{Combine} on $\mathbf{T}_B$ and $\mathbf{T}_K$, which involves her iterating through both the IBLTs and combining their entries in $O(\delta)$ time. Lastly, she performs \textsf{Peel} on the recovered IBLT $\mathbf{T}_L$. According to Eppstein {\it et al.} \cite{eppstein_whats_2011}, the peeling process for an IBLT of size $O(\delta)$ takes $O(\delta\log n)$ time. Thus, we get a total client time of $O(\delta\log n)$.

In $\mathsf{ServerAudit}$, the server performs both the server's and client's job in $\mathsf{ClientAudit}$, thus we combine the runtimes to get $O(\delta^2\log(n/\beta) + \delta\beta + \delta\log n)$. There is no communication or client work in $\mathsf{ServerAudit}$. 
\qed
\end{proof}

\section{IBLT Tree Construction} \label{sec:iblt-tree-construction}
\newcommand{\dtree}[1]{\overline{\mathcal{T}}_{#1}}

Suppose we have an implementation of a dynamic search tree with
\textsf{Init}$(S),$ \textsf{Insert}$(k, v, \dtree{S}),$
and \textsf{Delete}$(k, \dtree{S}),$
where $\dtree{S}$ is the dynamic search tree on set $S.$
Each of the IBLT tree operations can be implemented using these operations
as subroutines.

The \textsf{InitIbltTree} operation constructs $\tree{S}$ by first generating
$\dtree{S}$ with \textsf{Init}, then performing the following recursive
procedure to construct each node's IBLT, starting from the root node.
Given node $u,$ if $u$ is a leaf node, then the IBLT of $u$ is constructed by
inserting all $\Theta(\beta)$ blocks stored in $u.$
Otherwise, recursively construct the IBLTs of the children of $u$, then merge
all the IBLTs together and insert the block stored at $u.$

\begin{lemma}
  \label{lem:init-iblt-tree}
  Given $\dtree{S}$ with an implementation of \textsf{Init} which takes
  $T_\textsf{Init}$ time, the \textsf{InitIbltTree} operation takes
  $O(T_{\textsf{Init}} + L \beta + I c \delta)$ time,
  where $L$ and $I$ are the number of leaf and internal nodes respectively,
  and $c$ is the number of children per internal node.
\end{lemma}

\begin{proof}
  After $\dtree{S}$ is initialized, each leaf node must have its
  IBLT constructed, which takes $O(\beta)$ time per leaf.
  Then for each internal node, the cost of merging two IBLTs is $O(\delta)$
  time.
  Therefore merging across $c$ children takes $O(c \delta)$ time.
\qed
\end{proof}

Given a key $k$ and value $v$ to insert into $\tree{S},$
\textsf{InsertIbltTree} first performs \textsf{Insert} on $\dtree{S}.$
During the insertion, we track each node whose subtree changes such that their IBLT must change.
For each of these nodes, its IBLT is updated by merging its children,
similar to the \textsf{InitIbltTree} operation for $\tree{S}.$

\begin{lemma}
  \label{lem:insert-iblt-tree}
  Given $\dtree{S}$ with an implementation of \textsf{Insert} which takes
  $T_{\textsf{Insert}}$ time, the \textsf{InsertIbltTree} operation
  takes $O(T_{\textsf{Insert}} + \Delta_L \beta + \Delta_I c \delta)$ time,
  where $\Delta_L$ and $\Delta_I$ are the number of leaf and internal nodes
  which have different blocks prior to the insertion respectively,
  and $c$ is the number of children per internal node.
\end{lemma}

The previous lemma follows from the same reasoning as
Lemma~\ref{lem:init-iblt-tree}.
Note the \textsf{InsertIbltTree} operation allows the \textsf{Insert} operation
to handle the actual insertion of the block, and the IBLT tree simply updates
the IBLTs of the changed nodes accordingly.
Naturally the \textsf{DeleteIbltTree} operation works in the same way.

\begin{lemma}
  \label{lem:delete-iblt-tree}
  Given $\dtree{S}$ with an implementation of \textsf{Delete} which takes
  $T_{\textsf{Delete}}$ time, the \textsf{DeleteIbltTree} operation
  takes $O(T_{\textsf{Delete}} + \Delta_L \beta + \Delta_I c \delta)$ time,
  where $\Delta_L$ and $\Delta_I$ are the number of leaf and internal nodes
  which have different blocks prior to the deletion respectively,
  and $c$ is the number of children per internal node.
\end{lemma}

The \textsf{ConstructIBLT} operation constructs the IBLT of all blocks in
$\tree{S}$ excluding the set of keys $\{k^\prime_1, \dots, k^\prime_M\}$, through the
following recursive process.
Starting from the root node, check if the node $u$ contains any of the blocks
to exclude.
If it does not, then we can return the IBLT of the node.
If it does, and $u$ is a leaf node, then we construct the IBLT of the node
excluding any blocks we desire.
Otherwise, we recurse into each child of $u$, passing a segment of the
desired excluded blocks to each appropriate child.

\begin{lemma}
  \label{lem:construct-iblt}
  Given $\dtree{S},$ and a set of keys to exclude from the IBLT
  $\{k^\prime_1, \dots, k^\prime_M\},$ the \textsf{ConstructIBLT} operation takes
  $O(H M \delta + \delta \beta),$ where $H$ is the height of $\dtree{S}.$
\end{lemma}

\begin{proof}
  In the worst case, each of the $M$ blocks to exclude are in different
  leaf nodes.
  This would require $O(H)$ time to traverse $\dtree{S}$ to find the leaf.
  Then for each leaf node, we must construct the IBLT of the node excluding
  the desired blocks, which takes $O(\beta)$ time.
  Finally we have to recurse up the path from the leaf node to the root node,
  updating the IBLT of each node, which takes $O(\delta)$ time per node along
  the path.
  This gives $O(H \delta + \beta)$ per excluded block,
  or $O(H M \delta + \delta \beta)$ time in total.
\qed
\end{proof}

\begin{theorem}[Same as Theorem~\ref{thm:iblt-tree-cost}]
  Given a dynamic search tree with $\Theta(\beta)$ elements per leaf node,
  the IBLT tree can be implemented with $O(\delta)$ space overhead. 
  Given an implementation of the \textsf{Init}, \textsf{Insert}, and
  \textsf{Delete}, operations on the dynamic search tree, we can
  implement an IBLT tree with the following time complexities for each
  operation.
  Let
  $L$ and $I$ be the number of leaf and internal nodes respectively,
  $\Delta_L$  and $\Delta_I$ be the number of leaf and internal nodes
  which are changed by an insertion or deletion operation,
  $H$ be the height of the tree,
  $c$ be the number of children per internal node,
  and
  $M$ be the number of blocks to exclude when constructing the full IBLT.
  \begin{enumerate}
    \item The \textsf{InitIbltTree} operation takes 
      $O(T_\textsf{Init} + L\beta + I c \delta)$ time.
    \item The \textsf{InsertIbltTree} and \textsf{DeleteIbltTree} operations 
      take
      $O(T_\textsf{Insert} + \Delta_L \beta + \Delta_I c \delta)$
      and
      $O(T_\textsf{Delete} + \Delta_L \beta + \Delta_I c \delta)$
      time.
    \item The \textsf{ConstructIBLT} operation takes
      $O(H M \delta + \delta \beta)$ time.
  \end{enumerate}
\end{theorem}

\begin{proof}
  The theorem follows from Lemmas~\ref{lem:init-iblt-tree},
  \ref{lem:insert-iblt-tree}, \ref{lem:delete-iblt-tree},
  and
  \ref{lem:construct-iblt}.
\qed
\end{proof}

\subsection {Proof of Protected Cartesian Tree Construction}

The \textsf{InsertIbltTree} and \textsf{DeleteIbltTree} operations on the
the protected Cartesian tree can be implemented in a slightly more efficient
way in order to take advantage of the \textsf{Init} operation being
used as a subroutine in the \textsf{Insert} and \textsf{Delete} operations.
Note that while the \textsf{Insert} operation is the main focus,
the \textsf{Delete} operation is symmetric and can be analyzed in the same way.

The \textsf{Insert} operation on the protected Cartesian tree is implemented by
inserting the element into a node and observing whether the pivot at that node
needs to be updated.
If it does, then the entire subtree rooted at that node is reconstructed with
the \textsf{Init} operation.
Otherwise, the operation proceeds recursively in either the left or right
child if the inserted pair's key is less than or greater than the pivot
respectively.

Since the \textsf{Init} operation is called on at most one node $u$, for all
nodes on the path from the root to $u$ the IBLT of each $u$ can be updated in
$O(1)$ time by performing the \textsf{Update} operation on the IBLT of each
node.
Since this was a part of the cost of the \textsf{Insert} operation in the
protected Cartesian tree, by slightly modifying the analysis of the insertion
operation for the protected Cartesian tree to accomodate the
\textsf{InitIbltTree} operation, we can show Lemma \ref{lem:InsertIbltTree-runtime}.

To prove the lemma, we first introduce the following property of the protected
Cartesian tree, used in the proof.
\begin{lemma}[Adjusted from Bender et al.~\cite{bender}]
  \label{lem:insert-path-iblt-tree}
  Suppose the protected Cartesian tree after the insertion of $(k, v)$
  requires changes in the pivots in the tree.
  Then there is a pair of paths from the root to two leaf nodes which contain all of
  the pivot nodes which need to be updated.
\end{lemma}

Now we move on to the proof of Lemma \ref{lem:InsertIbltTree-runtime}.

\begin{proof}[of Lemma~\ref{lem:InsertIbltTree-runtime}]
  Let $z$ be the root of the protected Cartesian tree,
  and $S^\prime$ be the set $S \cup \{ (k, v) \}.$
  If $|U(S^\prime)| \leq \beta,$ then the bound holds trivially.
  Otherwise, the probability that the pivot of $z$ needs to be updated is
  $O(1 / n).$
  We first observe that the inserted pair $(k, v)$ either falls into the
  protected or unprotected region of $S$.
  In either region it falls, it must push another pair into the other region,
  i.e. $(k, v)$ falls into the protected region and another pair is pushed
  from the protected to the unprotected region.
  This gives us three cases in which the pivot of $z$ needs to be updated:
  1) $z$ is pushed from the protected to the unprotected region,
  2) the new element pushed into the unprotected region has a smaller $h$ value
     than the pivot of $z,$ or
  3) the new element has a smaller $h$ value than the pivot of $z.$
  All these independent events occur with probability $O(1 / |U(S^\prime)|),$
  and since $|U(S^\prime)| = n - \beta + 1 = \Theta(n),$ by definition of the
  protected region, therefore we get $(1 / n)$ probability that the pivot of
  $z$ needs to be updated.
  Since this would cause the entire tree rooted at $z$ to be reconstructed,
  which is of size $O(n / \beta),$ the expected time contribution from $z$
  is $O(\delta / \beta).$
  Lemma~\ref{lem:insert-path-iblt-tree} also states the number of nodes changed
  by the insertion can be found in at most two root to leaf paths.
  Using the same logic for $z$ for the other nodes along these two paths,
  the expected contribution of each of these nodes is $O(\delta / \beta).$
  Finally, since we know the number of nodes along these paths is $O(\log(n /
  \beta)),$ the expected time to perform the \textsf{InsertIbltTree} operation
  is $O(\delta \log(n / \beta) / \beta).$
\qed
\end{proof}

\end{document}